\documentclass[conference]{IEEEtran}
\IEEEoverridecommandlockouts

\usepackage{algorithm}
\usepackage{algpseudocode}
\usepackage{caption, subcaption, enumerate,varwidth}
\usepackage{wrapfig, epsfig, graphicx}
\usepackage{subcaption, array}
\usepackage{multirow}
\usepackage{booktabs, tabularx}
\usepackage[noadjust]{cite}
\usepackage{amsmath,amssymb,amsfonts,amsbsy,amsthm}
\usepackage{textcomp,color}
\definecolor{darkblue}{rgb}{0.0, 0.0, 0.55}
\usepackage[colorlinks=true,backref=page,urlcolor=darkblue,citecolor=blue,anchorcolor=blue,linkcolor=blue,draft]{hyperref}
\usepackage{mathtools}
\DeclarePairedDelimiter{\ceil}{\lceil}{\rceil}

\renewcommand*{\backref}[1]{}
\renewcommand*{\backrefalt}[4]{%
	\ifcase #1 (Not cited.)%
	\or        (page~#2)%
	\else      (pages~#2)%
	\fi}

\graphicspath{{figures/}{./}}

\newcommand{\ra}[1]{\renewcommand{\arraystretch}{#1}}
\newtheorem{theorem}{Theorem}

\renewenvironment{proof}{{\flushleft \textbf{Proof.}}}{\hfill $\blacksquare$}
\renewcommand\footnotemark{}

\def\BibTeX{{\rm B\kern-.05em{\sc i\kern-.025em b}\kern-.08em
    T\kern-.1667em\lower.7ex\hbox{E}\kern-.125emX}}
\begin{document}

\title{A Parallel Algorithm for Generating a Random Graph with a Prescribed Degree Sequence\\
}

\author{\IEEEauthorblockN{Hasanuzzaman Bhuiyan}
\IEEEauthorblockA{ \textit{Department of Computer Science}\\
\textit{Network Dynamics and Simulation }\\
\textit{Science Laboratory (NDSSL)}\\
\textit{Biocomplexity Institute of Virginia Tech}\\
Blacksburg, VA, USA \\
mhb@vt.edu}
\and
\IEEEauthorblockN{Maleq Khan}
\IEEEauthorblockA{\textit{Department of Electrical Engineering} \\
\textit{and Computer Science}\\
\textit{Texas A\&M University--Kingsville}\\
Kingsville, TX, USA \\
maleq.khan@tamuk.edu}
\and
\IEEEauthorblockN{Madhav Marathe}
\IEEEauthorblockA{ \textit{Department of Computer Science}\\
\textit{Network Dynamics and Simulation }\\
\textit{Science Laboratory (NDSSL)}\\
\textit{Biocomplexity Institute of Virginia Tech}\\
Blacksburg, VA, USA \\
mmarathe@vt.edu}
}

\maketitle
\begin{abstract}
Random graphs (or networks) have gained a significant increase of interest 
due to its popularity in modeling and simulating many complex
real-world systems.  Degree sequence is one of the most important aspects
of these systems.  Random graphs with a given degree sequence can capture
many characteristics like dependent edges and non-binomial degree distribution
that are absent in many classical random graph models such as the Erd\H{o}s-R\'{e}nyi 
graph model.  In addition, they have important applications in uniform sampling of
random graphs, counting the number of graphs having the same degree sequence,
as well as in string theory, random matrix theory, and matching theory.  
In this paper, we present an OpenMP-based shared-memory
parallel algorithm for generating a random graph with a prescribed degree
sequence, which achieves a speedup of 20.5 with 32 cores.
We also present a comparative study of several structural properties of
the random graphs generated by our algorithm with that of the real-world
graphs and random graphs generated by other popular methods.
One of the steps in our parallel algorithm requires checking the
Erd\H{o}s-Gallai characterization, i.e., whether there exists a graph
obeying the given degree sequence, in parallel.  This paper presents a 
non-trivial parallel algorithm for checking the Erd\H{o}s-Gallai
characterization, which achieves a speedup of 23 with 32 cores.
\end{abstract}


\section{Introduction}
\label{sec-introduction}
Random graphs are widely used for modeling many complex 
real-world systems such as
the Internet~\cite{siganos2003power}, biological~\cite{girvan2002community}, 
social~\cite{yang2015defining}, and infrastructure~\cite{latora2005vulnerability}
networks to understand how the systems work through
obtaining rigorous mathematical and simulation results.  Many random graph
models such as the Erd\H{o}s-R\'{e}nyi~\cite{erdosRenyi}, the Preferential 
Attachment~\cite{prefAttach}, the small-world~\cite{smallWorld}, and the
Chung-Lu~\cite{chung2002average} models have been proposed to capture various 
characteristics of real-world systems.  Degree sequence is one of the most
important aspects of these systems and has been extensively studied in graph
\mbox{theory~\cite{berge1973graphs,lovasz2009matching,hakimi1978graphs}.}  
It has significant applications in a wide range of areas including 
structural reliability and communication networks because of the strong
ties between the degrees of vertices and the structural properties
of and dynamics over a network~\cite{boccaletti2006complex}.

Random graphs with given degree sequences are widely used in uniform sampling
of random graphs as well as in counting the number of graphs having the same degree
sequence~\cite{cloteaux2016fast,blitzstein2011sequential,bayati2010sequential,del2010efficient}.
For example, in an epidemiology study of sexually transmitted diseases~\cite{liljeros2001web},
anonymous surveys collect data about the number of sexual partners of an individual
within a given period of time, and then the problem reduces to generating a network
obeying the degree sequence collected from the survey, and studying the disease
dynamics over the network.  Other examples include determining the total number
of structural isomers of chemical compounds such as alkanes, where the valence of
an atom is the degree.  Moreover, the random graphs with given 
degree sequences can capture many characteristics such as dependent edges
and non-binomial degree distribution that are absent in many classical models such as the
Erd\H{o}s-R\'{e}nyi~\cite{erdosRenyi} graph model.  They also have important
applications in string theory, random matrix theory, and matching 
theory~\cite{lovasz2009matching}.

The problem of generating a random graph with a given degree sequence becomes
considerably easier if self-loops and parallel edges are allowed.  Throughout
this paper, we consider simple graphs with no self-loops or parallel edges.
Most prior work on generating random graphs involves sequential algorithms,
and they can be broadly categorized in two classes: 
(\textbf{\textit{i}})~edge swapping and (\textbf{\textit{ii}})~stub-matching.
Edge \mbox{swapping~\cite{Feder,Gkantsidis,bhuiyan2014fast}}
uses the Markov chain Monte
Carlo~(MCMC) scheme on a given graph having the degree sequence.  An edge swap
operation replaces two edges $e_1=(a,b)$ and $e_2=(c,d)$, selected uniformly at
random from the graph, by new edges $e_3=(a,d)$ and $e_4=(c,b)$, 
i.e., the end vertices of the selected edges are swapped with each
other.  This operation is repeated either a given number of times or until a
specified criterion is satisfied.  It is easy to see that the degree of each
vertex remains invariant under an edge swap process.  Unfortunately, very
little theoretical results have been rigorously shown about the mixing
time~\cite{Feder,Cooper} of the edge swap process and they are ill-controlled.  
Moreover, most of the results are heuristic-based.

On the other hand, among the swap-free stub-matching methods, the \textit{configuration}
or \textit{pairing} method~\cite{Ne03} is very popular and uses a direct graph
construction method.  For each vertex, it creates as many stubs or ``dangling half-edges"
as of its degree.  Then edges are created by choosing pairs of vertices randomly and
connecting them.  This approach creates parallel edges, which are dealt with by restarting
the process.  Unfortunately, the probability of restarting the process
approaches~$1$ for larger degree sequences.  Many 
\mbox{variants~\cite{wormald99,steger99,kim04}} of the configuration models have been
studied to avoid parallel edges for the regular graphs.  By using the Havel-Hakimi
method~\cite{havelHakimi}, a deterministic graph can be generated following
a given degree sequence.
Bayati et al.~\cite{bayati2010sequential} presented an algorithm for counting
and generating random simple graphs with given degree sequences.  
However, this algorithm does not guarantee to always generate a graph, and it
is shown that the probability of not generating a graph is
small for a certain bound on the maximum degree, which restricts many degree
sequences.
Genio et al.~\cite{del2010efficient} presented an algorithm to generate a random
graph from a given degree sequence, which can be used in sampling graphs from
the graphical realizations of a degree sequence.  Blitzstein et 
al.~\cite{blitzstein2011sequential} also proposed a sequential importance
sampling~\cite{liu1998sequential} algorithm to generate random graphs with an
exact given degree sequence, which can generate every possible graph
with the given degree sequence with a non-zero probability.  Moreover, the
distribution of the generated graphs can be estimated, which is a much-desired
result used in sampling random graphs.

A deterministic parallel algorithm for generating a simple graph with a given
degree sequence has been presented by Arikati et al.~\cite{arikati1996realizing},
which runs in~$\mathcal{O}\left(\log n\right)$ time using~$\mathcal{O}\left(n+m\right)$
CRCW PRAM~\cite{karp1988survey} processors, where $n$ and $m$ denote the number of
vertices and edges in the graph, respectively.  From a given degree sequence, the
algorithm first computes an appropriate bipartite sequence (degree sequence of
a bipartite graph), generates a deterministic bipartite graph obeying the bipartite
sequence, applies some edge swap techniques to generate a symmetric bipartite graph,
and then reduces the symmetric bipartite graph to a simple graph having the
given degree sequence.  
Another parallel algorithm, with a time complexity of~$\mathcal{O}\left(\log^4 n\right)$
using~$\mathcal{O}\left(n^{10}\right)$ EREW PRAM processors, has been presented
in~\cite{dessmark1994parallel}, where the maximum degree is bounded by the
square-root of the sum of the degrees, which restricts many degree sequences.  
A parallel algorithm for generating
a random graph with a given \textit{expected} degree sequence has been presented
in~\cite{alam2015parallel}.  However, there is no existing parallel algorithm for
generating random graphs following an exact degree sequence, which can provably generate
each possible graph, having the given degree sequence, with a positive probability.
In this paper, we present an efficient parallel algorithm for generating a
random graph with an exact given degree sequence.  We choose to parallelize the
sequential algorithm by Blitzstein et al.~\cite{blitzstein2011sequential} because
of its rigorous mathematical and theoretical results, and the algorithm supports
all of the important and much-desired properties below, whereas the other algorithms
are either heuristic-based or lack some of the following properties:
\begin{itemize}
	\item It can construct a random simple graph with a prescribed degree sequence.
	\item It can provably generate each possible graph, obeying the given degree
	sequence, with a positive probability.
	\item It can be used in importance sampling by explicitly measuring
	the weights associated with the generated graphs.
	\item It is guaranteed to terminate with a graph having the prescribed
	degree sequence.
	\item Given a degree sequence of a tree, a small tweak
	while assigning the edges allows the same algorithm to generate 
	trees uniformly at random.
	\item It can be used in estimating the number of possible graphs 
	with the given degree sequence.
\end{itemize}

\textbf{Our Contributions.}
In this paper, we present an efficient shared-memory parallel algorithm
for generating random graphs with exact given degree sequences.  
The dependencies among
assigning edges to vertices in a particular order to ensure the algorithm always
successfully terminates with a graph, the requirement of keeping the graph simple,
maintaining an exact stochastic process as that of the sequential algorithm, 
and concurrent writing by multiple cores in the global address space lead to significant
challenges in designing a parallel algorithm.  Dealing with these requires complex
synchronization among the processing cores.  Our parallel algorithm achieves a
maximum speedup of $20.5$ with $32$ cores.  We also present a comparative study 
of various structural properties of the random graphs generated by the parallel
algorithm with that of the real-world graphs.  One of the steps in our parallel
algorithm requires checking the graphicality of a given degree sequence, i.e.,
whether there exists a graph with the degree sequence, using the Erd\H{o}s-Gallai
characterization~\cite{erdHos1961grafok} in parallel.  
We present here a novel parallel algorithm for checking the Erd\H{o}s-Gallai
characterization, which achieves a speedup of $23$ using $32$ cores.

\textbf{Organization.}
The rest of the paper is organized as follows.  Section~\ref{sec:preliminaries}
describes the preliminaries and notations used in the paper.  Our main
parallel algorithm for generating random graphs along with the experimental
results are presented in Section~\ref{sec:generate-random-network}.  
We present a parallel algorithm for checking the Erd\H{o}s-Gallai characterization
of a given degree sequence accompanied by the performance evaluation of the algorithm
in Section~\ref{sec:erdos-gallai}.  Finally, we conclude in Section~\ref{sec:conclusion}.

\section{Preliminaries}
\label{sec:preliminaries}
Below are the notations, definitions, and computation model
used in this paper.

\textbf{Notations.}
We use $\mathbb{G}=(\mathbb{V},\mathbb{E})$ to denote a simple graph, where
$\mathbb{V}$ is the set of vertices and $\mathbb{E}$ is the set of edges.
A \textit{self-loop} is an edge from a vertex to itself. \textit{Parallel edges}
are two or more edges connecting the same pair of vertices.  A \textit{simple
graph} is an undirected graph with no self-loops or parallel edges.  
We are given a \textit{degree sequence} $\mathbb{D} = (d_1,d_2, \ldots, d_n)$. 
There are a total of $n=|\mathbb{V}|$ vertices labeled as $1,2, \ldots, n,$ and
$d_i$ is the degree of vertex $i$, where $0 \le d_i \le n-1$. 
For a degree sequence $\mathbb{D}$ and distinct $u,v \in \{1,2,\ldots,n\}$,
we define $\ominus_{u,v}^\mathbb{D}$ to be the degree sequence obtained
from $\mathbb{D}$ by subtracting 1 from each of $d_u$ and $d_v$.  
Let $d'_j$ be the degree of vertex $j$ in the
degree sequence $\ominus_{u,v}^\mathbb{D}$, then
\begin{equation}
d'_j =
\begin{cases}
d_j -1  & \text{ if } j \in \{u,v\}\text{,}\\
d_j  	& \text{ otherwise.}\\
\end{cases}
\end{equation}
If there is a simple graph $\mathbb{G}$ having the degree sequence $\mathbb{D}$,
then there are $m=|\mathbb{E}|$ edges in $\mathbb{G}$, where $2m = \sum_i d_i$.
The terms graph and network are used interchangeably throughout the paper.  
We use \texttt{K}, \texttt{M}, and \texttt{B} to denote
thousands, millions, and billions, respectively; e.g., \texttt{1M} stands
for one million. For the parallel algorithms, let $\mathcal{P}$ be the number
of processing cores, and $\mathbb{P}_k$ the core with rank~$k$,
where $0 \le k < \mathcal{P}$.  
A summary of the frequently used notations (some of them are introduced
later for convenience) is provided in Table~\ref{table:notation}.

\begin{table}[tb!]
	\caption{Notations used frequently in the paper.} \label{table:notation}
	\ra{1.3}
	\centering
	\begin{tabular}{@{}llll@{}}
		\toprule[1.3pt]
		\multicolumn{1}{@{}l@{}}{\textbf{Symbol}} &
		\multicolumn{1}{@{}l@{}}{\textbf{~~Description}} &
		\multicolumn{1}{@{}l@{}}{\textbf{~Symbol}} &
		\multicolumn{1}{@{}l@{}}{\textbf{~~Description}} 
		\\
		\midrule[0.8pt]
		$\mathbb{D}$ & Degree sequence & $d_i$ & Degree of vertex $i$ \\
		$\mathbb{V}$ & Set of vertices & $n$ & Number of vertices \\
		$\mathbb{E}$ & Set of edges & $m$ & Number of edges \\
		$\mathcal{P}$ & Number of cores & $\mathbb{P}_k$ & Core with rank $k$ \\
		$\mathbb{C}$ & Candidate set & $\mathcal{C}$ & Corrected Durfee number \\
		$\mathbb{G}$ & Graph & \texttt{K} & Thousands \\
		\texttt{M} & Millions & \texttt{B} & Billions \\
		\bottomrule[1.3pt]
	\end{tabular}
	\vspace{-0.15in}
\end{table}

\textbf{Residual Degree.}
During the course of a graph generation process,
the \textit{residual degree} of a vertex $u$ is the remaining number
of edges incident on $u$, which have not
been created yet.  From hereon, we refer to the degree $d_u$ of a vertex $u$ as 
the residual degree of $u$ at any given time, unless otherwise specified.

\textbf{Graphical Sequence.} A degree sequence $\mathbb{D}$ of non-negative 
integers is called \textit{graphical} if there exists a labeled simple graph with
vertex set $\{1,2,\ldots,n\}$, where vertex~$i$ has degree~$d_i$.  Such
a graph is called a \textit{realization} of the degree sequence~$\mathbb{D}$.
Note that there can be several graphs having the same degree sequence.
Eight equivalent necessary and sufficient conditions for testing the graphicality
of a degree sequence are listed in~\cite{mahadev1995threshold}. Among them, the
Erd\H{o}s-Gallai characterization~\cite{erdHos1961grafok} is the most famous
and frequently used criterion.  Another popular recursive test for checking a
graphical sequence is the Havel-Hakimi method~\cite{havelHakimi}.

\textbf{Erd\H{o}s-Gallai Characterization~\cite{erdHos1961grafok}.}  
Assuming a given degree sequence $\mathbb{D}$ is sorted in non-increasing order, 
i.e., $d_1 \ge d_2 \ge \ldots \ge d_n$, the sequence $\mathbb{D}$ is
graphical if and only if $\sum_{i=1}^{n} d_i$ is even and
\begin{equation} \label{eqn:erdos-gallai}
	\textnormal{for each } k\in \{1,2,\ldots,n\}, \sum_{i=1}^{k}d_i \le k(k-1) + \sum_{i=k+1}^{n} \min(k,d_i).
\end{equation}
For example, $\mathbb{D}_1 = (3,3,2,2,2)$ is a graphical sequence and there is a 
realization of $\mathbb{D}_1$ as it satisfies the Erd\H{o}s-Gallai characterization,
whereas $\mathbb{D}_2 = (4,3,2,1)$ is not a graphical sequence and there is no
simple graph realizing $\mathbb{D}_2$, as shown in Figs.~\ref{fig:erdos-gallai}
and~\ref{fig:realization}.

\begin{figure}[htb]
	\begin{minipage}[t]{.32\textwidth}
		\begin{center}
		\centerline{\includegraphics[width=1.0\textwidth]{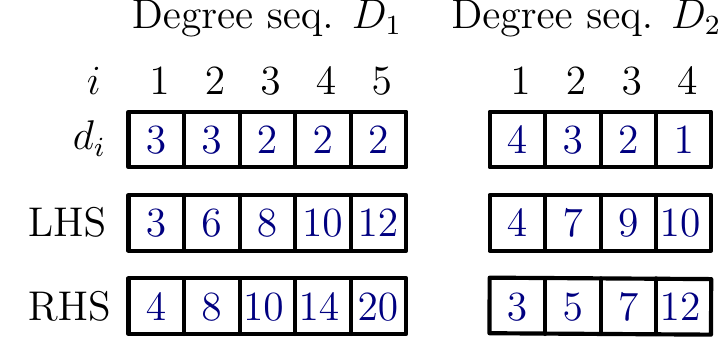}}
		\caption{Graphicality check for the degree sequences
			$\mathbb{D}_1 = (3,3,2,2,2)$ and $\mathbb{D}_2 = (4,3,2,1)$ 
			using the Erd\H{o}s-Gallai characterization, where LHS and RHS
			denote the left hand side and right hand side values of 
			Eq.~(\ref{eqn:erdos-gallai}), respectively.}
		\label{fig:erdos-gallai}
		\end{center}
	\end{minipage}
	\hfill
	\begin{minipage}[t]{.14\textwidth}
		\begin{center}
			\centerline{\includegraphics[width=1.0\textwidth]{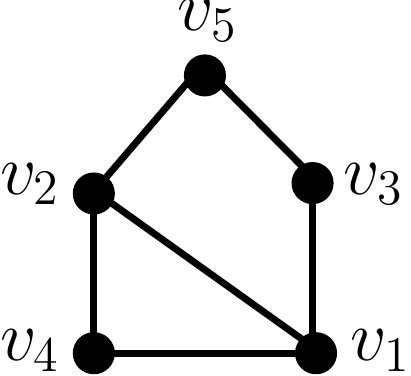}}
			\caption{A simple graph realizing the degree sequence $\mathbb{D}_1 = (3,3,2,2,2)$.}
			\label{fig:realization}
		\end{center}
	\end{minipage}
	\hfill
	\vspace{-0.15in}
\end{figure}

\textbf{Computation Model.}  We develop algorithms for
shared-memory parallel systems.  All the cores can read from and write to
the global address space.  In addition, each core can have its own local
variables and data structures. 
\section{Generating Random Graphs with Prescribed Degree Sequences}
\label{sec:generate-random-network}
We briefly discuss the sequential algorithm in Section~\ref{sec:seq-algo-graph-generate}.
Then we present our parallel algorithm in Section~\ref{sec:par-algo-graph-generate}
and the experimental results in Section~\ref{sec:experiment-graph-generate}.

\subsection{Sequential Algorithm}
\label{sec:seq-algo-graph-generate}
Blitzstein et al.~\cite{blitzstein2011sequential} presented a sequential
importance sampling~\cite{liu1998sequential} algorithm for generating random
graphs with exact prescribed degree sequences.  This approach first creates
all edges incident on the vertex having the minimum degree in the sequence,
then moves to the next vertex having the minimum degree to create its incident
edges and so on.  To create an edge incident on a vertex $u$, a candidate list
$\mathbb{C}$ is computed using the Erd\H{o}s-Gallai characterization such
that, after adding an edge by connecting $u$ to any candidate vertex $v$
from the list $\mathbb{C}$, the residual degree sequence remains graphical 
and the graph remains simple. 
Then an edge $(u,v)$ is assigned by choosing $v$ from the candidate list 
$\mathbb{C}$ with a probability proportional to the degree of $v$.  This 
process is repeated until all edges incident on vertex $u$ are assigned.

For example, for a given degree sequence $\mathbb{D}=(3,3,2,2,2)$, the algorithm
starts by assigning edges incident on vertex $v_3$.  It computes the candidate list
$\mathbb{C} = \{v_1,v_2,v_4,v_5\}$.  Say it chooses the vertex~$v_5$ from~$\mathbb{C}$
and assigns the edge~$(v_3,v_5)$.  Then the new degree sequence is $\mathbb{D}=(3,3,1,2,1)$,
and the new candidate list for assigning the remaining edge incident on vertex~$v_3$ is
$\mathbb{C} = \{v_1,v_2\}$.  Say the algorithm selects~$v_1$ from~$\mathbb{C}$ and
assigns the edge~$(v_3,v_1)$.  Now the new degree sequence is $\mathbb{D}=(2,3,0,2,1)$,
and the algorithm will proceed to assign edges incident on vertex~$v_5$ and so on.
One possible sequence of degree sequences is 
\begin{equation*} \label{eqn:deg-seq}
	\begin{split}
		(3,3,2,2,2) \rightarrow (3,3,1,2,1) \rightarrow (2,3,0,2,1) \rightarrow (2,2,0,2,0) \\
		\rightarrow (1,2,0,1,0) \rightarrow (0,1,0,1,0) \rightarrow (0,0,0,0,0),
	\end{split}
\end{equation*}
with the corresponding edge set
\begin{equation*} \label{eqn:edge-set}
	\mathbb{E} = \{(v_3,v_5), (v_3,v_1), (v_5,v_2), (v_1,v_4), (v_1,v_2), (v_2,v_4)\}.
\end{equation*}

\begin{center}
	\begin{minipage}{0.48\textwidth}
		\small{
			\noindent\fbox{%
				\begin{minipage}{\dimexpr0.98\linewidth-2\fboxsep-2\fboxrule\relax}
					\begin{algorithmic}[1]
						\State $\mathbb{E} \leftarrow \emptyset$ \hspace{0.1in} \texttt{$\triangleright$ initially empty set of edges}
						\State \textbf{while} $\mathbb{D} \ne \textbf{0}$ \textbf{do}
						\State \hspace{0.1in} Select the least $u$ such that $d_u$ is a minimal positive degree in $\mathbb{D}$
						\State \hspace{0.1in} \textbf{while} $d_u \ne 0$ \textbf{do} 
						\State \hspace{0.2in} $\mathbb{C} \leftarrow \{ v \ne u: (u,v) \notin \mathbb{E}~\bigwedge~ \ominus_{u,v}^\mathbb{D}$ is graphical$\}$ \label{seq-gen:a}
						\State \hspace{0.2in} $v \leftarrow$ a random candidate in $\mathbb{C}$ where probability of selecting $v$ is proportional to $d_v$ \label{seq-gen:b}
						\State \hspace{0.2in} $\mathbb{E} \leftarrow \mathbb{E} \cup \{(u,v)\} $
						\State \hspace{0.2in} $\mathbb{D} \leftarrow \ominus_{u,v}^\mathbb{D}$
						\State Output $\mathbb{E}$
					\end{algorithmic}
				\end{minipage}
			}
		}
		\captionof{figure}{A sequential algorithm~\cite{blitzstein2011sequential} 
			for generating a random graph with a given degree sequence.}
		\label{algo:seq-graph-generation-algo}
	\end{minipage}
\end{center}
The corresponding graph is shown in Fig.~\ref{fig:realization}.
Note that during the assignment of incident edges on a vertex $u$, a candidate
at a later stage is also a candidate at an earlier stage.  The pseudocode of the
algorithm is shown in Fig.~\ref{algo:seq-graph-generation-algo}.  Since a total
of $m$ edges are generated for the graph $\mathbb{G}$ and computing the candidate
list (Line~\ref{seq-gen:a}) for each edge takes $\mathcal{O} \left ( n^2\right )$
time, the time complexity of the algorithm is $\mathcal{O} \left ( mn^2\right )$.

Unlike many other graph generation algorithms, this method never gets stuck,
i.e., it always terminates with a graph realizing the given degree sequence
(proof provided in Theorem~$3$ in~\cite{blitzstein2011sequential}) or creates
loops or parallel edges through the computation of the candidate list using the 
Erd\H{o}s-Gallai characterization.  The algorithm can generate every possible
graph with a positive probability (proof given in Corollary~$1$
in~\cite{blitzstein2011sequential}).  For additional details about the importance
sampling and estimating the number of graphs for a given degree sequence, see
Sections~$8$ and~$9$ in~\cite{blitzstein2011sequential}; and we 
omit the details in this paper due to space constraints.

%
%
\subsection{Parallel Algorithm}
\label{sec:par-algo-graph-generate}
To design an exact parallel version by maintaining the same stochastic process 
(in order to retain the same theoretical and mathematical results) as that
of the sequential algorithm, the vertices are considered (to assign their
incident edges) in the same order in the parallel algorithm, i.e., in
ascending order of their degrees.  Hence, we emphasize parallelizing the
computation of the candidate list~$\mathbb{C}$, i.e., Line~\ref{seq-gen:a} of
the sequential algorithm in Fig.~\ref{algo:seq-graph-generation-algo}.  
For computing the candidate list to assign edges incident on a vertex~$u$,
we need to consider all other vertices~$v$ with non-zero degrees~$d_v$ as potential
candidates; and we parallelize this step.  While considering a particular
vertex~$v$ as a candidate, we need to check whether~$\ominus_{u,v}^\mathbb{D}$
is a graphical sequence using the Erd\H{o}s-Gallai characterization.  
If~$\ominus_{u,v}^\mathbb{D}$ is graphical, then~$v$ is added to the candidate
list~$\mathbb{C}$.  The time complexity of the best known sequential algorithm
for testing the Erd\H{o}s-Gallai characterization is~$\mathcal{O} \left ( n \right 
)$~\cite{ivanyi2011erdos,blitzstein2011sequential}.  Thus to have an efficient
parallel algorithm for generating random graphs, we need to use an efficient
parallel algorithm for checking the Erd\H{o}s-Gallai characterization.
In Section~\ref{sec:erdos-gallai}, we present an efficient
parallel algorithm for checking the Erd\H{o}s-Gallai characterization that
runs in~$\mathcal{O} \left ( \frac{n}{\mathcal{P}} + \log \mathcal{P} \right )$
time.  The parallel algorithm for the Erd\H{o}s-Gallai characterization returns
\textit{TRUE} if the given degree sequence is graphical and \textit{FALSE} otherwise.

\begin{center}
	\begin{minipage}{0.48\textwidth}
		\small{
			\noindent\fbox{%
				\begin{minipage}{\dimexpr0.98\linewidth-2\fboxsep-2\fboxrule\relax}
					\begin{algorithmic}[1]
						
						\State $\mathbb{E} \leftarrow \emptyset$ \hspace{0.1in} \texttt{$\triangleright$ initially empty set of edges}
						
						\vspace{0.1in}
						\Statex \texttt{$\triangleright$ Assign the edges until the degree of each vertex reduces to 0}
						\State  \textbf{while} $\mathbb{D} \ne \textbf{0}$ \textbf{do}		
						\State  \hspace{0.1in} Select the least $u$ such that $d_u$ is a minimal positive degree in $\mathbb{D}$ \label{PRGG:p}
						\State  \hspace{0.1in} $\mathbb{C} \leftarrow \emptyset$ \hspace{1.1in}\texttt{$\triangleright$ candidate list}
						
						\vspace{0.1in}
						\Statex \hspace{0.1in}  \texttt{$\triangleright$ Assign all $d_u$ edges incident on $u$}
						\State  \hspace{0.1in}  \textbf{while} $d_u \ne 0$ \textbf{do} \label{PRGG:q}
						\State  \hspace{0.25in} \textbf{if} $\mathbb{C} = \emptyset$ \textbf{then} \label{PRGG:w}
						\State  \hspace{0.4in}  $\mathbb{F} \leftarrow \{ v \ne u: (u,v) \notin \mathbb{E}~\bigwedge~d_v > 0\}$ \label{PRGG:r}
						\State  \hspace{0.25in} \textbf{else} 
						\State  \hspace{0.4in}  $\mathbb{F} \leftarrow \mathbb{C}$ \label{PRGG:s}
						\State  \hspace{0.4in}  $\mathbb{C} \leftarrow \emptyset$ \label{PRGG:t}
						
						\vspace{0.1in}
						\Statex \hspace{0.25in} \texttt{$\triangleright$ Compute the candidate list}
						\State  \hspace{0.25in} \textbf{for} \textit{each} $v \in \mathbb{F}$ \textbf{\textit{in parallel} do} \label{PRGG:e}
						\State  \hspace{0.4in}  $flag \leftarrow $ \textsc{Parallel-Erd\H{o}s-Gallai} $\left ( \ominus_{u,v}^\mathbb{D} \right )$ \label{PRGG:f}
						\State  \hspace{0.4in}  \textbf{if} $flag = TRUE$ \textbf{then}
						\State  \hspace{0.55in} $\mathbb{C} \leftarrow \mathbb{C} \cup \{v\}$ \hspace{0.2in}\texttt{$\triangleright$ $v$ is a candidate} \label{PRGG:u}		
						
						\vspace{0.1in}
						\State \hspace{0.25in} \textbf{if} $d_u = |\mathbb{C}|$ \textbf{then} \label{PRGG:a}
						\State \hspace{0.4in} \textbf{for} \textit{each} $v \in \mathbb{C}$ \textbf{\textit{in parallel} do}
						\State \hspace{0.55in} $\mathbb{E} \leftarrow \mathbb{E} \cup \{(u,v)\}$ 
						\State \hspace{0.55in} $\mathbb{D} \leftarrow \ominus_{u,v}^\mathbb{D}$
						\State \hspace{0.4in} \textbf{break} \label{PRGG:b}
						
						\vspace{0.1in}
						\Statex \hspace{0.25in} \texttt{$\triangleright$ Assign an edge $(u,v)$ from $\mathbb{C}$}
						\State  \hspace{0.25in} $v \leftarrow$ a random candidate in $\mathbb{C}$ where probability of selecting $v$ is proportional to $d_v$ \label{PRGG:c}
						\State  \hspace{0.25in} $\mathbb{E} \leftarrow \mathbb{E} \cup \{(u,v)\}$ \label{PRGG:d}
						\State  \hspace{0.25in} $\mathbb{D} \leftarrow \ominus_{u,v}^\mathbb{D}$ \label{PRGG:g}
						\State  \hspace{0.25in} $\mathbb{C} \leftarrow \mathbb{C} - \{v\}$  \label{PRGG:h}
						
						\vspace{0.1in}
						\State Output $\mathbb{E}$ \hspace{0.3in} \texttt{$\triangleright$ final set of edges}
					\end{algorithmic}
				\end{minipage}
			}
		}
		\captionof{figure}{A parallel algorithm for generating a random graph
			with a prescribed degree sequence.}
		\label{algo:par-graph-generation-algo}
	\end{minipage}
\end{center}

Once the candidate list is computed, if the degree of~$u$ is equal to the
cardinality~$|\mathbb{C}|$ of the candidate list, then new edges are assigned
between~$u$ and all candidate vertices~$v$ in
the candidate list~$\mathbb{C}$ in parallel. 
Otherwise, like the sequential algorithm, a candidate vertex~$v$ is chosen randomly
from~$\mathbb{C}$, a new edge~$(u,v)$ is assigned, the degree sequence~$\mathbb{D}$
is updated by reducing the degrees of each of~$u$ and~$v$ by~$1$, and the process is
repeated until~$d_u$ is reduced to $0$.  After assigning all edges
incident on vertex $u$, the algorithm proceeds with assigning edges incident
on the next vertex having the minimum positive degree in $\mathbb{D}$ and so on.
We present the pseudocode of our parallel algorithm for generating random
graphs in Fig.~\ref{algo:par-graph-generation-algo}.

\begin{figure*}[htb!]
	\hfill
	\begin{minipage}[t]{.475\textwidth}
		\begin{center}
			\centerline{\includegraphics[width=1.0\textwidth]{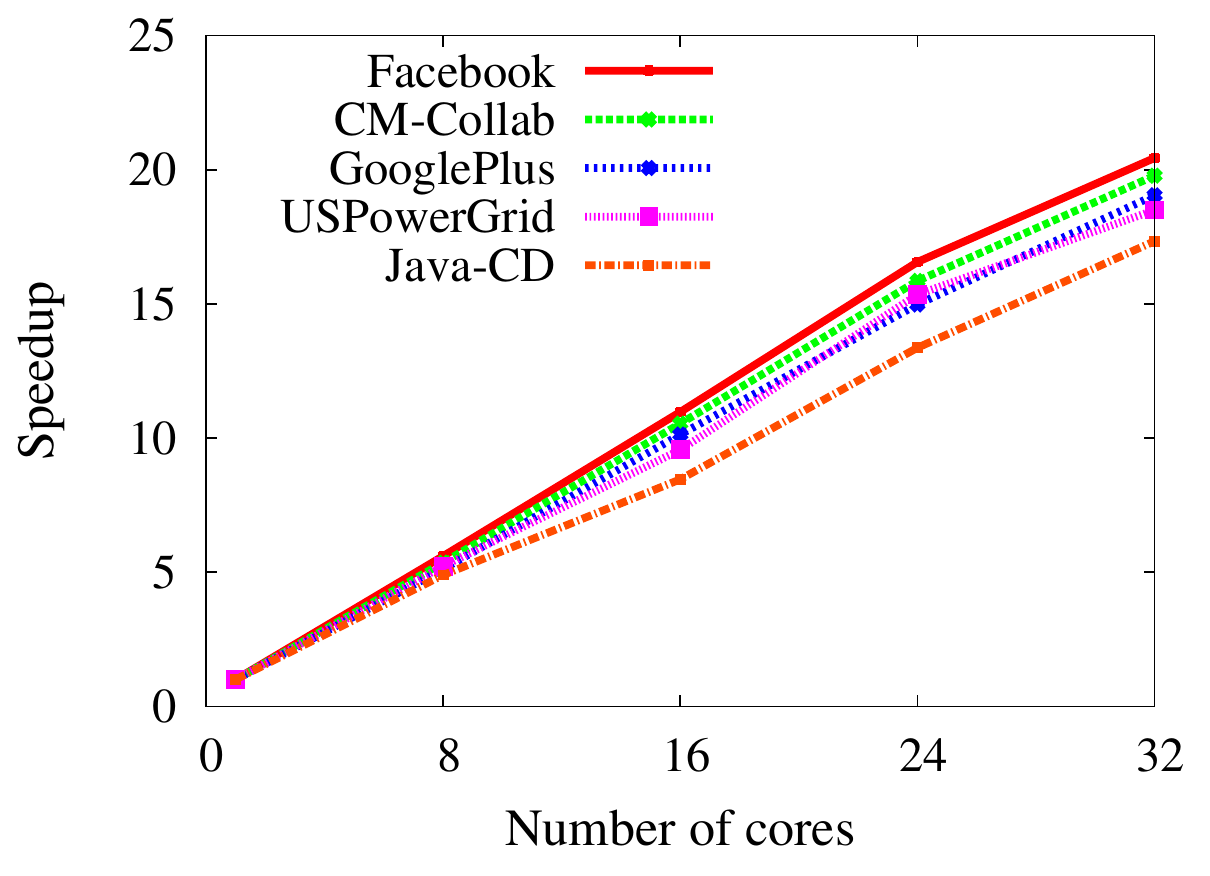}}
			\caption{Strong scaling of our parallel algorithm for generating
				random graphs on different data sets.}
			\label{fig:strong_scaling_graph_generate}
		\end{center}
	\end{minipage}
	\hfill
	\begin{minipage}[t]{.475\textwidth}
		\begin{center}
			\centerline{\includegraphics[width=1.0\textwidth]{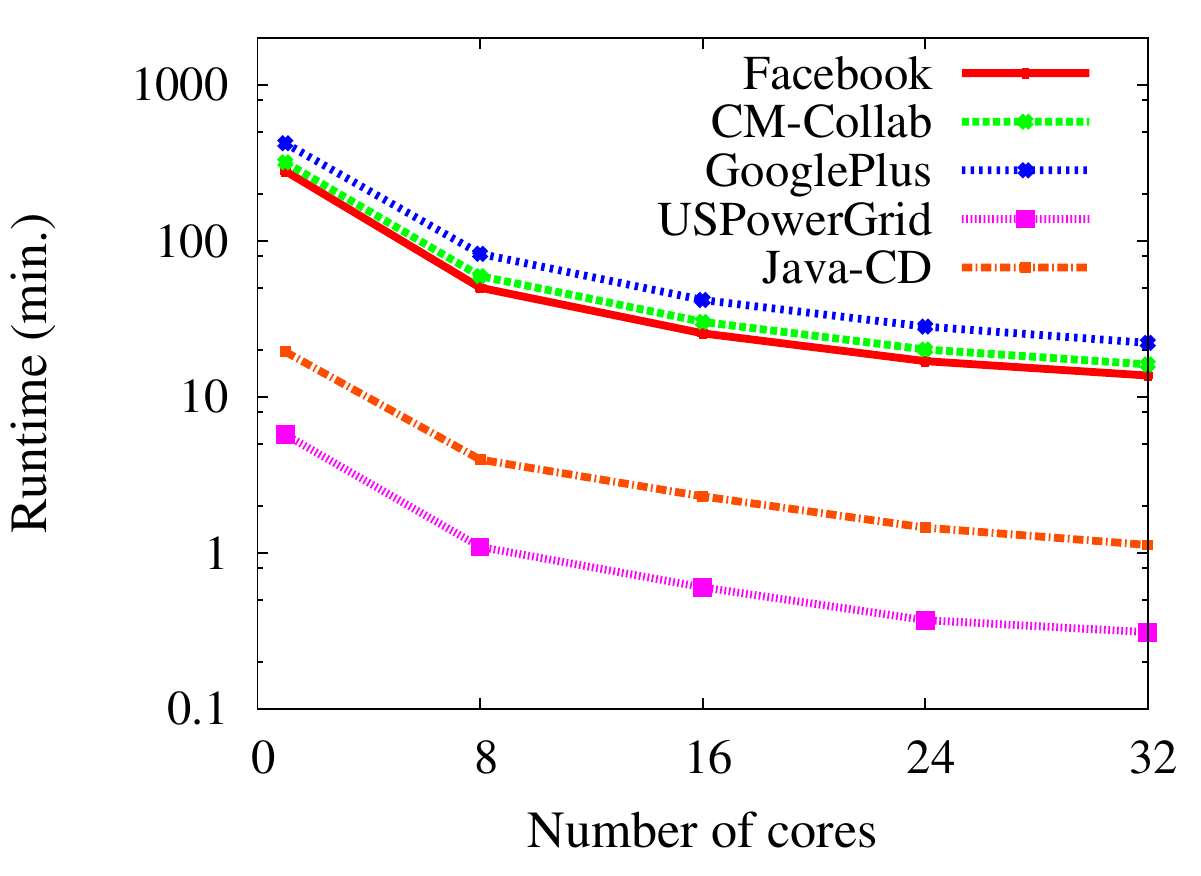}}
			\caption{Runtime of our parallel algorithm for generating
				random graphs on different data sets.}
			\label{fig:runtime_graph_generate}
		\end{center}
	\end{minipage}
	\hfill
	\vspace{-0.15in}
\end{figure*}

\begin{theorem}
	The parallel algorithm for generating random graphs maintains an exact
	stochastic process as that of the sequential algorithm and preserves all
	mathematical and theoretical results of the sequential algorithm.
\end{theorem}
\begin{proof}
	The parallel algorithm always selects the vertex $u$ with the minimum degree
	in the sequence (Line~\ref{PRGG:p}), assigns~$d_u$ edges incident on~$u$ 
	(Lines~\ref{PRGG:q}-\ref{PRGG:h}), and then proceeds with the next vertex
	in the sequence as the sequential algorithm would do.  While assigning the first
	edge incident on a vertex~$u$, all vertices in the sequence that do not create
	self-loops or parallel edges are considered as potential candidates~$\mathbb{F}$
	(Line~\ref{PRGG:r}), whereas for assigning the subsequent edges incident
	on~$u$, the candidates~$\mathbb{C}$ in an earlier stage are considered as
	the potential candidates~$\mathbb{F}$ (Line~\ref{PRGG:s}) in a later stage.
	The candidate list is then computed in parallel by checking whether an edge
	can be assigned between~$u$ and a potential candidate~$v$ in~$\mathbb{F}$
	by checking whether the residual degree sequence~$\ominus_{u,v}^\mathbb{D}$,
	if an edge $(u,v)$ is assigned,	is a graphical sequence by using the parallel
	algorithm for the Erd\H{o}s-Gallai characterization (Lines~\ref{PRGG:e}-\ref{PRGG:u}).  
	If the cardinality of the candidate list is equal to the degree $d_u$ of vertex
	$u$, then edges are assigned between~$u$ and all vertices~$v$ in the candidate
	list~$\mathbb{C}$ in parallel (Lines~\ref{PRGG:a}-\ref{PRGG:b}).  Although this
	step is not explicitly mentioned in the sequential algorithm, this is obvious
	since the sequential algorithm would assign all $d_u$ edges incident on $u$
	and there are no additional candidates other than the $d_u$ candidates 
	in~$\mathbb{C}$.  We parallelize this step to improve the performance of the
	algorithm.  If the candidate list~$\mathbb{C}$ has more than~$d_u$ candidates,
	then a vertex~$v$ is selected randomly from~$\mathbb{C}$ with probability
	proportional to~$d_v$, and an edge~$(u,v)$ is assigned, as the sequential
	algorithm would do.  Hence, the parallel algorithm maintains an exact
	stochastic process as that of the sequential algorithm.  As a consequence,
	all mathematical and theoretical results (except the time complexity) of the
	sequential algorithm are applicable to our parallel algorithm as well.
\end{proof}
\begin{theorem}
	The time complexity of each of the core $\mathbb{P}_k$ 
	in the parallel algorithm for generating random graphs is
	$\mathcal{O} \left ( mn \left( \frac{n}{\mathcal{P}} + \log \mathcal{P} \right ) \right )$.
\end{theorem}
\begin{proof}
	The parallel algorithm assigns $m$ edges one by one.  To assign an edge
	incident on a vertex $u$, it computes the candidate list in parallel (Line~\ref{PRGG:e}).  
	Whether a vertex $v$ is a candidate is computed using the parallel algorithm
	for the Erd\H{o}s-Gallai characterization (Line~\ref{PRGG:f}), which has a time complexity of
	$\mathcal{O} \left (\frac{n}{\mathcal{P}} + \log \mathcal{P} \right )$.  
	Hence, the time complexity of the parallel algorithm for generating random
	graphs is $\mathcal{O} \left ( mn \left(\frac{n}{\mathcal{P}} + \log \mathcal{P}
	\right ) \right )$.
\end{proof}

\begin{theorem}
	The space complexity of the parallel algorithm for generating random graphs is
	$\mathcal{O} \left ( m+n \right )$.
\end{theorem}
\begin{proof}
	Storing the degree sequence	and the edges take $\mathcal{O}
	\left ( n \right )$ and $\mathcal{O}\left ( m \right )$ space, respectively, making
	a space requirement of $\mathcal{O} \left ( m+n \right )$.
\end{proof}

\subsection{Experimental Results}
\label{sec:experiment-graph-generate}
In this section, we present the data sets used in the experiments and
the strong scaling and runtime of our parallel algorithm for generating
random graphs.

\textbf{Experimental Setup.}
We use a $32$-core Haswell-EP E$5$-$2698$ v$3$ $2.30$GHz ($3.60$GHz Turbo) dual processor 
node with $128$GB of memory, $1$TB internal hard drive, and QLogic QDR InfiniBand adapter.
We use OpenMP version $3.1$ and GCC version $4.7.2$ for implementation.

\textbf{Data Sets.}  We use degree sequences of five real-world
networks for the experiments.  A summary of the networks is given in 
Table~\ref{table:data-graph-generate}.  Facebook~\cite{nr-sigkdd16} is an
anonymized Facebook friendship network of the students of CMU.
GooglePlus~\cite{nr-sigkdd16} is an online social contact network of GooglePlus.
The USPowerGrid~\cite{abdelhamid2012cinet} network represents a high-voltage
power grid in the western states of the USA. 
Java-CD~\cite{abdelhamid2012cinet} is a Java class dependency network of
JUNG 2.0.1~\cite{jung}.  CM-Collab~\cite{newman2001structure} is a scientific
collaboration network on the condensed matter topic.

\begin{table}[tb!]
	\caption{Data sets used in the experiments, where $n$, $m$, and $\frac{2m}{n}$ 
		denote the no. of vertices, no. of edges, and average degree
		of the networks, respectively.  \texttt{K} denotes thousands.}
	\label{table:data-graph-generate}
	\ra{1.3}
	\centering
	\begin{tabular}{@{}llccc@{}}
		\toprule[1.3pt]
		\multicolumn{1}{@{}l@{}}{\textbf{Network}} &
		\multicolumn{1}{@{}l@{}}{\textbf{~~Type}} &
		\multicolumn{1}{@{}c@{}}{$n$} &
		\multicolumn{1}{@{}c@{}}{$m$} &
		\multicolumn{1}{@{}c@{}}{$\frac{2m}{n}$}
		\\
		\midrule[0.8pt]
		Facebook~\cite{nr-sigkdd16} & Social contact & \texttt{6.6K} & \texttt{250K} & \texttt{75.50} \\
		GooglePlus~\cite{nr-sigkdd16} & Social contact & \texttt{23.6K} & \texttt{39.2K} & \texttt{3.32} \\
		USPowerGrid~\cite{abdelhamid2012cinet} & Power grid & \texttt{4.94K} & \texttt{6.6K} & \texttt{2.67} \\
		Java-CD~\cite{abdelhamid2012cinet} & Dependency & \texttt{6.12K} & \texttt{50.3K} & \texttt{16.43} \\
		CM-Collab~\cite{newman2001structure} & Collaboration & \texttt{16.3K} & \texttt{47.6K} & \texttt{5.85} \\
		\bottomrule[1.3pt]
	\end{tabular}
	\vspace{-0.15in}
\end{table}

\begin{table*}[htb!]
	\caption{A comparison of some structural properties of the random networks
		generated (from the degree sequences of the	real-world networks) by our
		parallel algorithm  with that of the real-world networks and random networks
		generated by swapping 100\% edges of the real-world networks.  
		We use average values of 20 experiments.}
	\label{table:properties}
	\ra{1.3}
	\centering
	\begin{tabular}{@{}llcccccccc@{}}
		\toprule
		& & \multicolumn{8}{c}{\textbf{Network structural properties}} \\
		\cmidrule{3-10}
		&&&&&&& \multicolumn{3}{c}{\textbf{Average vertex value}}\\
		\cmidrule{8-10}
		\textbf{Network}& \textbf{Network} &\textbf{Triangles}&\textbf{Cliques}&\textbf{Connected}& \textbf{Avg. shortest} & \textbf{Diameter} & \textbf{Betweenness} & \textbf{Closeness} & \textbf{Clustering} \\
		&\textbf{model}&&&\textbf{component}& \textbf{path length} && \textbf{centr. (x$10^{-4}$)} & \textbf{centrality} & \textbf{coefficient} \\
		\midrule
		\multirow{3}{*}{Facebook} & Real-world 			    &\texttt{2.31M}&\texttt{1.24M}&\texttt{1}&\texttt{2.74}&\texttt{8}&\texttt{2.63}&\texttt{0.37}&\texttt{0.28} \\
		& Our algo. &\texttt{0.57M}&\texttt{0.40M}&\texttt{1}&\texttt{2.50}&\texttt{6}&\texttt{2.27}&\texttt{0.40}&\texttt{0.04} \\
		& Edge swap &\texttt{0.54M}&\texttt{0.39M}&\texttt{1}&\texttt{2.49}&\texttt{5}&\texttt{2.26}&\texttt{0.41}&\texttt{0.04} \\
		\cmidrule{1-10}
		\multirow{3}{*}{GooglePlus} & Real-world &\texttt{18.22K}&\texttt{31.09K}&\texttt{4}&\texttt{4.03}&\texttt{8}&\texttt{1.28}&\texttt{0.25}&\texttt{0.17} \\
		& Our algo. &\texttt{163.7K}&\texttt{21.96K}&\texttt{1.6K}&\texttt{3.20}&\texttt{5}&\texttt{0.69}&\texttt{0.24}&\texttt{0.22} \\
		& Edge swap &\texttt{99.95K}&\texttt{1.27M}&\texttt{637}&\texttt{3.13}&\texttt{9}&\texttt{0.81}&\texttt{0.29}&\texttt{0.19} \\
		\cmidrule{1-10}
		\multirow{3}{*}{USPowerGrid} & Real-world &\texttt{651}&\texttt{5.69K}&\texttt{1}&\texttt{18.99}&\texttt{46}&\texttt{36.42}&\texttt{0.05}&\texttt{0.08} \\
		& Our algo. &\texttt{8}&\texttt{6.58K}&\texttt{74}&\texttt{8.48}&\texttt{20}&\texttt{14.07}&\texttt{0.11}&\texttt{0.0008} \\
		& Edge swap &\texttt{2}&\texttt{6.59K}&\texttt{88}&\texttt{8.49}&\texttt{22}&\texttt{13.91}&\texttt{0.11}& \texttt{0.0003}\\
		\cmidrule{1-10}
		\multirow{3}{*}{Java-CD} & Real-world &\texttt{0.18M}&\texttt{31.89K}&\texttt{1}&\texttt{2.11}&\texttt{7}&\texttt{1.82}&\texttt{0.48}&\texttt{0.68} \\
		& Our algo. &\texttt{0.29M}&\texttt{21.34K}&\texttt{1}&\texttt{2.10}&\texttt{5}&\texttt{1.80}&\texttt{0.48}&\texttt{0.66} \\
		& Edge swap &\texttt{0.19M}&\texttt{55.79K}&\texttt{1}&\texttt{2.00}&\texttt{4}&\texttt{1.64}&\texttt{0.50}&\texttt{0.69} \\
		\cmidrule{1-10}
		\multirow{3}{*}{CM-Collab} & Real-world &\texttt{68K}&\texttt{10.49K}&\texttt{726}&\texttt{6.63}&\texttt{18}&\texttt{2.51}&\texttt{0.11}&\texttt{0.64} \\
		& Our algo. &\texttt{272}&\texttt{47.09K}&\texttt{25}&\texttt{4.91}&\texttt{12}&\texttt{2.39}&\texttt{0.20}&\texttt{0.0012} \\
		& Edge swap &\texttt{264}&\texttt{47.09K}&\texttt{31}&\texttt{4.91}&\texttt{14}&\texttt{2.39}&\texttt{0.20}&\texttt{0.0010} \\
		\bottomrule
	\end{tabular}
	\vspace{-0.1in}
\end{table*}

\begin{figure}[htb!]
	\begin{center}
		\includegraphics[width=0.475\textwidth]{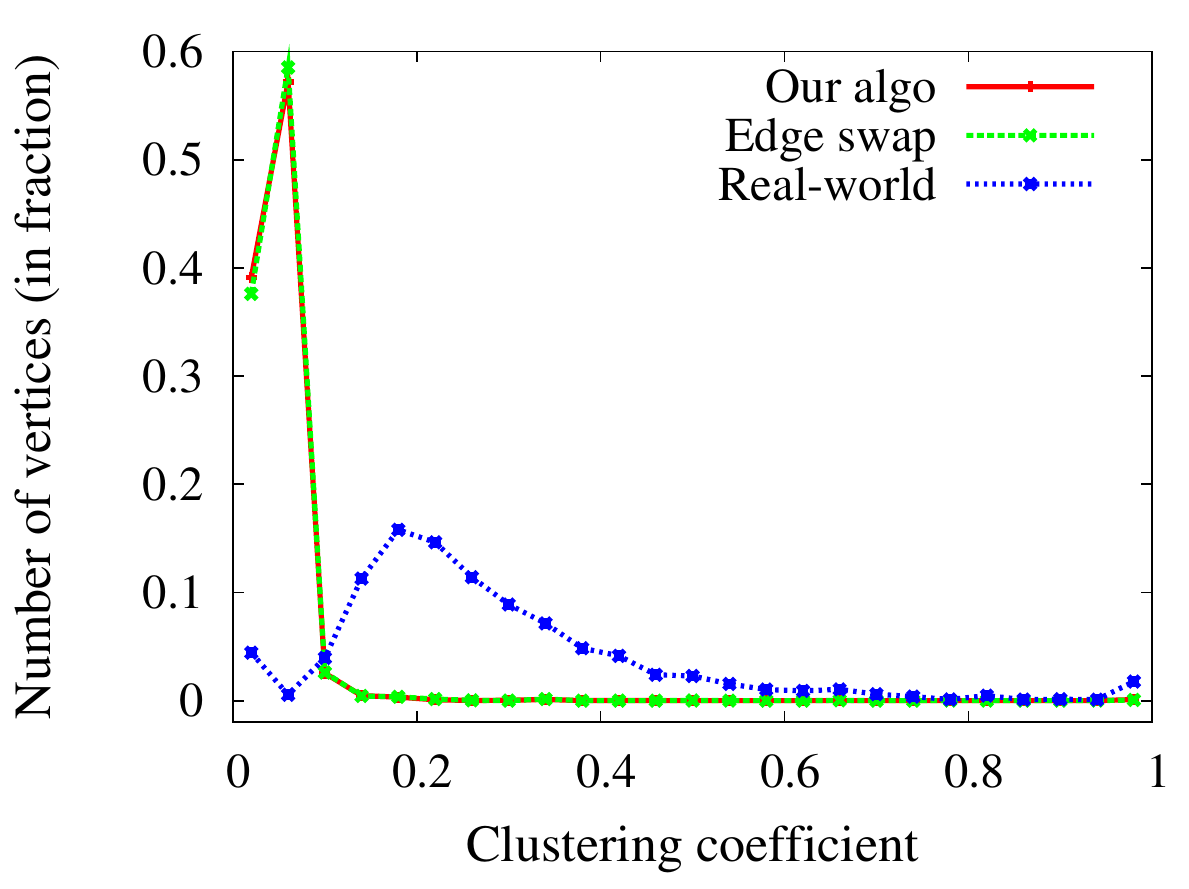}
		\caption{A comparison of clustering coefficient distributions of the real-world
			Facebook network and random networks generated by our parallel algorithm and
			by swapping edges of the real-world network.  The distributions of the random
			networks almost completely overlap with each other.}
		\label{fig:facebook-cc-comparison}
	\end{center}
	\vspace{-0.2in}
\end{figure}

\textbf{Strong Scaling.}  The strong scaling and runtime of the parallel algorithm are
shown in Figs.~\ref{fig:strong_scaling_graph_generate}~and~\ref{fig:runtime_graph_generate},
respectively.  The speedups increase almost linearly with the increase in the
number of cores, and we achieve a maximum speedup of $20.5$ with $32$ cores on
the Facebook network.  

\textbf{Structural Properties of the Generated Graphs.}
A comparison of the structural properties of the random graphs generated by our parallel
algorithm with that of the real-world graphs and random graphs generated by swapping
edges is given in Table~\ref{table:properties}.  To generate random graphs by swapping
edges, we perform $\frac{m}{2} \ln m$ edge swap operations to swap $100\%$ edges of the
real-world graphs~\cite{bhuiyan2017parallel}.  We use average values of $20$ experiments.
We study the number of triangles, cliques, connected components, average shortest path length,
diameter, average betweenness centrality, average closeness centrality, and average local 
clustering coefficient of the networks.  We observe that in many cases the properties of
the random graphs are far away from the real-world graphs.  Moreover, the structural properties
of the random graphs generated by our algorithm and by swapping edges are very close to each
other in most of the cases.  For example, the clustering coefficient distributions, as shown
in Fig.~\ref{fig:facebook-cc-comparison}, of the random graphs generated by these two methods
almost completely overlap with each other, and it is difficult to distinguish them in the
figure, whereas both of them lie far away from the real-world graph.

%
%
\section{Parallel Algorithm for Checking the Erd\H{o}s-Gallai Characterization}
\label{sec:erdos-gallai}
Many variants of the Erd\H{o}s-Gallai characterization have been developed
and proofs have been given (see~\cite{blitzstein2011sequential} 
and~\cite{mahadev1995threshold} for a good discussion).  Such a useful result
has been presented in Theorem $3.4.1$ in~\cite{mahadev1995threshold}, which
defines the \textit{corrected Durfee number} $\mathcal{C}$ of the degree
sequence $\mathbb{D} = (d_1,d_2, \ldots, d_n)$ (sorted in non-increasing order) as
\begin{equation}
	\mathcal{C} = \left | \left \{ j : d_j \ge j-1 \right \} \right |
\end{equation}
and showed that $\mathbb{D}$ is graphical if and only if it satisfies the first
$\mathcal{C}$ inequalities of the Erd\H{o}s-Gallai test.  The corrected Durfee
number $\mathcal{C}$ is often significantly smaller than the number of vertices
$n$.  For example, for the degree sequence \mbox{$\mathbb{D} = (3,2,2,2,1)$,} 
the corrected Durfee number $\mathcal{C}$ is $3$, as shown in Fig.~\ref{fig:durfeeno};
hence, it is sufficient to check only the first three Erd\H{o}s-Gallai inequalities
instead of checking all five inequalities of Eq.~(\ref{eqn:erdos-gallai}).
\begin{figure}[tb!]
	\centerline{\includegraphics[width=0.65\linewidth]{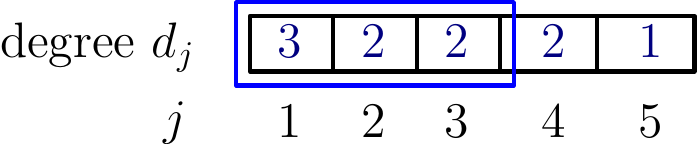}}
	\caption{For the degree sequence $\mathbb{D} = (3,2,2,2,1)$,
		the corrected Durfee number $\mathcal{C} = 3$.}
	\label{fig:durfeeno}
	\vspace{-0.1in}
\end{figure}

The sequential algorithm for checking the Erd\H{o}s-Gallai characterization
is quite straightforward and has a time complexity of
$\mathcal{O}\left (n \right)$~\cite{ivanyi2011erdos}.  A parallel algorithm for
the problem has been presented
in~\cite{dessmark1994parallel}, which has a runtime of~$\mathcal{O}\left(\log
n\right)$ using~$\mathcal{O}\left(\frac{n}{\log n}\right)$ EREW PRAM processors.  
In addition to generating random graphs, the Erd\H{o}s-Gallai characterization
has important applications in many other graph theory problems as well.  For example,
Iv{\'a}nyi
et al.~\cite{ivanyi2011parallel} applied the sequential algorithm for checking the
Erd\H{o}s-Gallai characterization to enumerate the distinct degree sequences of
simple graphs in
\begin{center}
	\begin{minipage}{0.48\textwidth}
		\small{
			\noindent\fbox{%
				\begin{minipage}{\dimexpr0.98\linewidth-2\fboxsep-2\fboxrule\relax}
					\begin{algorithmic}[1]
						\Statex \texttt{$\triangleright$ Compute the corrected Durfee number}
						\State $i \leftarrow 1$
						\State \textbf{while} $i <= n$ \textit{and} $d_i \ge i-1$ \textbf{do}
						\State \hspace{0.1in} $\mathcal{C} \leftarrow i$ \hspace{0.58in}\texttt{// corrected Durfee number}
						\State \hspace{0.1in} $i \leftarrow i + 1$
						
						\vspace{0.1in}			
						\Statex \texttt{$\triangleright$ Compute the prefix sum of the degrees}
						\State $H_0 \leftarrow 0$
						\State \textbf{for} $i = 1$ \textit{to} $n$ \textbf{do}
						\State \hspace{0.1in} $H_i \leftarrow H_{i-1} + d_i$
						
						\vspace{0.1in}						
						\Statex \texttt{$\triangleright$ Check the parity}
						\State \textbf{if} $H_n$ \textit{is} odd \textbf{then}
						\State \hspace{0.1in} \textbf{return} \textit{FALSE} \hspace{0.1in}\texttt{// not a graphical sequence}
						
						\vspace{0.1in}
						\Statex \texttt{$\triangleright$ Compute the weights}
						\State $d_0 \leftarrow n-1$
						\State \textbf{for} $i=1$ \textit{to} $n$ \textbf{do}
						\State \hspace{0.1in} \textbf{if} $d_i < d_{i-1}$ \textbf{then}
						\State \hspace{0.25in} \textbf{for} $j=d_{i-1}$ \textit{downto} $d_i+1$ \textbf{do}
						\State \hspace{0.4in} $w_j \leftarrow i-1$
						\State \hspace{0.25in} $w_{d_i} \leftarrow i$						
						\State \textbf{for} $j=d_n$ \textit{downto} $1$ \textbf{do}
						\State \hspace{0.1in} $w_j \leftarrow n$						
						
						\vspace{0.1in}
						\Statex \texttt{$\triangleright$ Check the Erd}\H{o}\texttt{s-Gallai inequalities}
						\State \textbf{for} $i=1$ \textit{to} $\mathcal{C}$ \textbf{do}
						\State \hspace{0.1in} \textbf{if} $i \le w_{i}$ \textbf{then}
						\State \hspace{0.25in} \textbf{if} $H_i > i(i-1) + i(w_i-i) + H_n - H_{w_i}$ \textbf{then}
						\State \hspace{0.4in} \textbf{return} \textit{FALSE}
						\State \hspace{0.1in} \textbf{else if} $H_i > i(i-1) + H_n - H_i$ \textbf{then}
						\State \hspace{0.25in} \textbf{return} \textit{FALSE}
						
						\vspace{0.1in}
						\State \textbf{return} \textit{TRUE} \hspace{0.58in}\texttt{// a graphical sequence}
					\end{algorithmic}
				\end{minipage}
			}
		}
		\captionof{figure}{A sequential algorithm~\cite{ivanyi2011erdos} for checking the Erd\H{o}s-Gallai characterization.}
		\label{algo:seq-erdos-gallai}
	\end{minipage}
\end{center}
parallel.  In this section, we present a shared-memory parallel algorithm
for checking the Erd\H{o}s-Gallai characterization of a given degree
sequence with a time complexity 
of~$\mathcal{O} \left ( \frac{n}{\mathcal{P}} + \log \mathcal{P} \right )$
using~$\mathcal{P}$ processing cores.  
First we briefly review the current state-of-the-art sequential algorithm.

\subsection{Sequential Algorithm}
The sequential algorithm~\cite{ivanyi2011erdos} is quite simple and
consists of the following steps: (\textit{\textbf{i}})~compute the
corrected Durfee number~$\mathcal{C}$,
(\textit{\textbf{ii}})~compute
the prefix sum of the degrees, (\textit{\textbf{iii}})~check the
parity, i.e., whether the
sum of the degrees is even or odd, (\textit{\textbf{iv}})~compute the weights,
which are useful in computing the right hand side of Eq.~(\ref{eqn:erdos-gallai})
in linear time, and (\textit{\textbf{v}})~check the first~$\mathcal{C}$
Erd\H{o}s-Gallai inequalities.  If the sum of the degrees is even and all
the inequalities are satisfied, then the degree sequence~$\mathbb{D}$ is graphical;
otherwise,~$\mathbb{D}$ is not graphical.  The pseudocode of the sequential algorithm
is given in Fig.~\ref{algo:seq-erdos-gallai}.
\subsection{Parallel Algorithm}
Based on the sequential algorithm presented in Fig.~\ref{algo:seq-erdos-gallai},
we present a parallel algorithm for checking the Erd\H{o}s-Gallai characterization.
Below we describe the methodology to parallelize the steps of
the sequential algorithm.

$\bullet$ \textbf{\textit{Step 1}: Compute the Corrected Durfee Number.}
The corrected Durfee number can be computed in parallel in a round robin
fashion, as shown in Fig.~\ref{algo:par-durfee}.  Each core
$\mathbb{P}_k$ computes its local corrected Durfee number $\mathcal{C}_k$.
Then all the cores synchronize and the maximum value of all 
$\mathcal{C}_k$ is reduced as the corrected Durfee number $\mathcal{C}$.

$\bullet$ \textbf{\textit{Step 2}: Compute the Prefix Sum of the Degrees.}
We use a parallel version~\cite{aluru2012teaching} of computing the prefix sum,
as shown in Fig.~\ref{algo:par-prefix-sum}.  Each core $\mathbb{P}_k$ works on
a chunk of size $\ceil[\big]{\frac{n}{\mathcal{P}}}$ of the 

\begin{center}
	\begin{minipage}{0.48\textwidth}
		\small{
			\noindent\fbox{%
				\begin{minipage}{\dimexpr0.98\linewidth-2\fboxsep-2\fboxrule\relax}
					\begin{algorithmic}[1]
						\State $k \leftarrow$ core id
						\vspace{0.1in}
						\Statex \texttt{$\triangleright$ Each core $\mathbb{P}_k$ executes the following in parallel:}
						\State $i \leftarrow k+1$
						\State $\mathcal{C}_k \leftarrow 0$
						\State \textbf{while} $i <= n$ \textit{and} $d_i \ge i-1$ \textbf{do} \label{PDURF:a}
						\State \hspace{0.1in} $\mathcal{C}_k \leftarrow i$ \hspace{0.1in}\texttt{// local corrected Durfee number}
						\State \hspace{0.1in} $i \leftarrow i + \mathcal{P}$ \label{PDURF:b}
						
						\vspace{0.1in}
						\Statex \texttt{$\triangleright$ Reduce the corrected Durfee number}
						\State $\mathcal{C} \leftarrow \textsc{Reduce-Max}_{k}$ $\mathcal{C}_k$ \label{PDURF:c}
					\end{algorithmic}
				\end{minipage}
			}
		}
		\captionof{figure}{Compute the corrected Durfee number in parallel.}
		\label{algo:par-durfee}
	\end{minipage}
\end{center}
\begin{center}
	\begin{minipage}{0.48\textwidth}
		\small{
			\noindent\fbox{%
				\begin{minipage}{\dimexpr0.98\linewidth-2\fboxsep-2\fboxrule\relax}
					\begin{algorithmic}[1]
						\State $k \leftarrow$ core id
						\vspace{0.1in}
						\Statex \texttt{$\triangleright$ Each core $\mathbb{P}_k$ executes the following in parallel:}
						\vspace{0.01in}\State $x \leftarrow k \ceil[\big]{\frac{n}{\mathcal{P}}} + 1$
						\vspace{0.01in}\State $y \leftarrow \min \left \{(k+1) \ceil[\big]{\frac{n}{\mathcal{P}}}, n \right \}$
						\vspace{0.01in}\State $s_k \leftarrow \sum_{i=x}^{y} d_i$ \label{PEG:a}
						\vspace{0.01in}\State \textbf{In Parallel:} $ S_k \leftarrow \sum_{j=0}^{k-1} s_j$ \label{PEG:b}
						\State $Q \leftarrow S_k$ \hspace{1.04in} \texttt{// note that $S_0 = 0$} \label{PEG:c}
						\State \textbf{for} $i=x$ \textit{to} $y$ \textbf{do} \label{PEG:e}
						\State \hspace{0.1in} $H_i \leftarrow Q + d_i$
						\State \hspace{0.1in} $Q \leftarrow H_i$ \label{PEG:d}
					\end{algorithmic}
				\end{minipage}
			}
		}
		\captionof{figure}{Compute the prefix sum of the degrees in parallel.}
		\label{algo:par-prefix-sum}
	\end{minipage}
\end{center}
degree sequence.  First, the sum $s_k$ of the degrees in the chunk is
computed (Line~\ref{PEG:a}) and then a prefix sum~$S_k$
of the $s_j$ ($0 \le j \le k-1$) is
computed in parallel (Line~\ref{PEG:b}).  Finally, each core gives a pass to the chunk
and uses the value of $S_k$ to compute the final prefix sum (Lines~\ref{PEG:c}-\ref{PEG:d}).

$\bullet$ \textbf{\textit{Step 3}: Check the Parity.}
The master core checks whether the sum of the degrees
is even.  If the sum is odd, then the degree sequence is not graphical.
Otherwise, the algorithm proceeds to the next step.

%
%

%
$\bullet$ \textbf{\textit{Step 4}: Compute the Weights.}
The pseudocode of computing the weights in parallel is shown in Fig.~\ref{algo:par-weight}.
We first initialize (Lines~\ref{PWT:a}-\ref{PWT:b}) the weight array $w$ in parallel.  
Then the actual weights are computed inside a \texttt{for} loop (Lines~\ref{PWT:c}-\ref{PWT:d})
in parallel.  Due to the simultaneous nature of the parallel algorithm, there is a possibility
that the same weight $w_j$ may be updated by multiple cores in an order different than that
of the sequential algorithm.  To deal with this difficulty, we add two additional
\texttt{if} conditions (Lines~\ref{PWT:e} and~\ref{PWT:f}) as the values of $w_j$ are
only updated with larger values in the sequential algorithm.  These two conditions ensure
the correctness of the weight values as well as allow simultaneous parallel computation
of them.  Finally, the larger weights are computed in parallel in the last \texttt{for}
loop (Lines~\ref{PWT:g}-\ref{PWT:h}).

$\bullet$ \textbf{\textit{Step 5}: Check the Erd\H{o}s-Gallai Inequalities.}
The Erd\H{o}s-Gallai inequalities can be checked in parallel in a round
robin fashion.  We have to check only the first $\mathcal{C}$ inequalities
instead of checking all the $n$ inequalities.  This significantly improves
the performance of the algorithm since $\mathcal{C} << n$ in many degree
sequences, as shown later in Table~\ref{table:data-erdos-gallai}.
If any of the inequalities is dissatisfied, then the degree sequence is not graphical;
otherwise, it is a graphical sequence.  The pseudocode of the algorithm is presented
in Fig.~\ref{algo:par-eg-inequality}.

\begin{figure*}[htb!]
	\hfill
	\begin{minipage}[t]{.475\textwidth}
		\begin{center}
			\centerline{\includegraphics[width=1.0\textwidth]{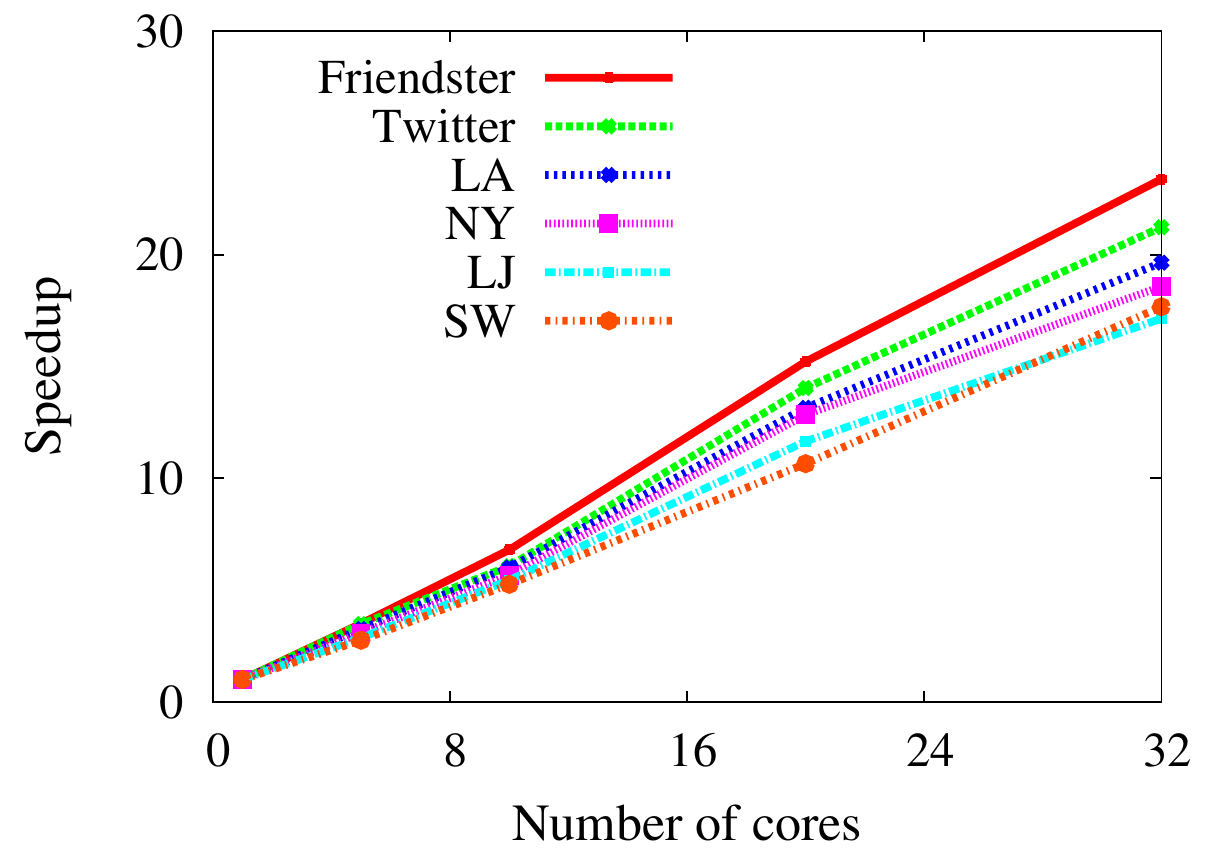}}
			\caption{Strong scaling of the parallel algorithm for checking the
				Erd\H{o}s-Gallai characterization.}
			\label{fig:strong_scaling_erdos_gallai}
		\end{center}
	\end{minipage}
	\hfill
	\begin{minipage}[t]{.475\textwidth}
		\begin{center}
			\centerline{\includegraphics[width=1.0\textwidth]{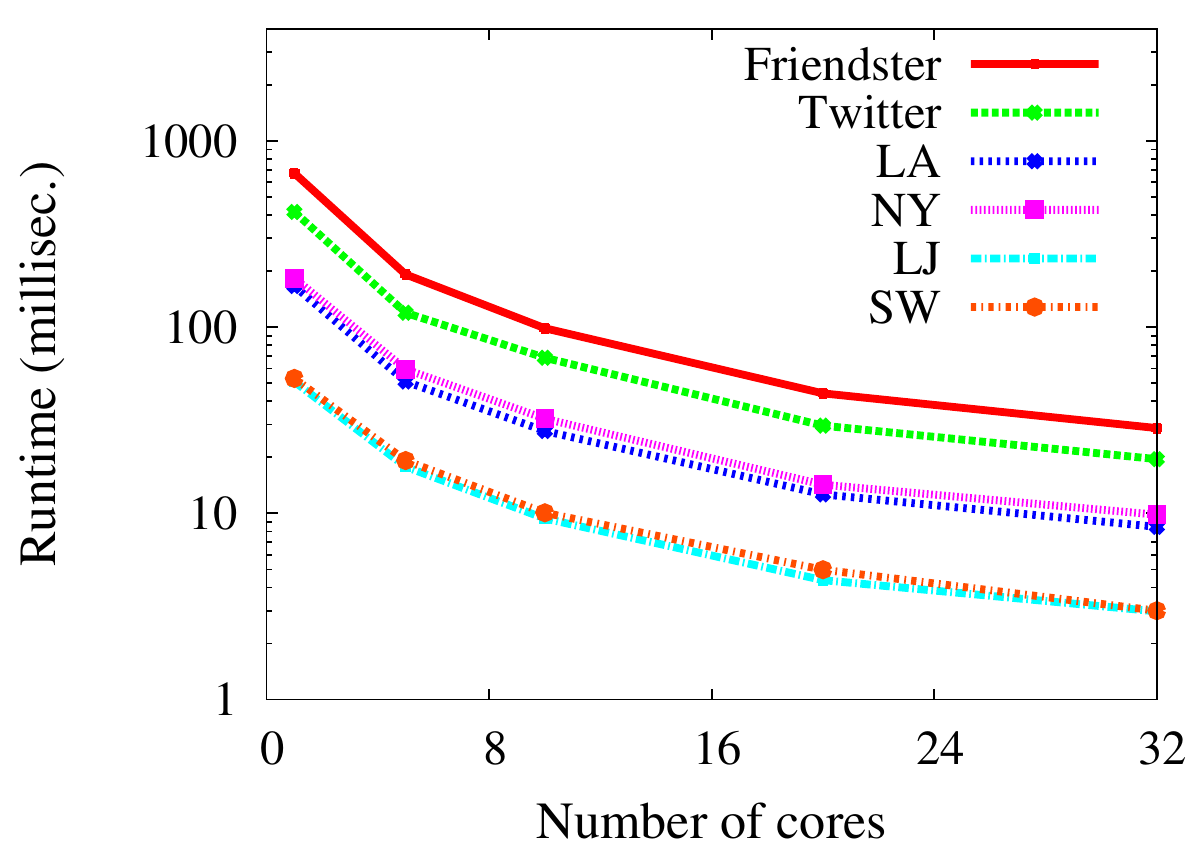}}
			\caption{Runtime of the parallel algorithm for checking the
				Erd\H{o}s-Gallai characterization.}
			\label{fig:runtime_erdos_gallai}
		\end{center}
	\end{minipage}
	\hfill
	\vspace{-0.2in}
\end{figure*}

\begin{theorem}
	The time complexity of each of the core $\mathbb{P}_k$ 
	in the parallel algorithm for checking the Erd\H{o}s-Gallai
	characterization is $\mathcal{O} \left ( \frac{n}{\mathcal{P}} + \log \mathcal{P} \right )$.
\end{theorem}
\begin{proof}
	The \texttt{while} loop
	in Lines~\ref{PDURF:a}-\ref{PDURF:b} in Fig.~\ref{algo:par-durfee} takes $\mathcal{O} \left 
	( \frac{\mathcal{C}}{\mathcal{P}} \right )$ time, where $\mathcal{C}$ is the
	corrected Durfee number and $\mathcal{P}$ is the number of cores.  
	The \texttt{reduction} in Line~\ref{PDURF:c} takes $\mathcal{O} \left ( \log \mathcal{P}
	\right )$ time.  Hence, the corrected Durfee
	number in Step $1$ can be computed in $\mathcal{O} \left ( \frac{\mathcal{C}}{\mathcal{P}}
	+ \log \mathcal{P} \right )$ time.  Lines~\ref{PEG:a},~\ref{PEG:b}, and~\ref{PEG:e}-\ref{PEG:d} in 
	Fig.~\ref{algo:par-prefix-sum} take $\mathcal{O} \left ( \frac{n}{\mathcal{P}} \right )$, 
	$\mathcal{O} \left (\log \mathcal{P} \right )$, and $\mathcal{O} \left ( \frac{n}{\mathcal{P}} 
	\right )$ time, respectively, where $n$ is the number of vertices.  So, the prefix sum of the 
	degrees in Step $2$ can be computed in $\mathcal{O} \left ( \frac{n}{\mathcal{P}} +
	\log \mathcal{P} \right )$ time~\cite{aluru2012teaching}.  Checking the parity in Step
	$3$ can be done in $\mathcal{O}(1)$ time.  Each of the three \texttt{for} loops 
	(Lines~\ref{PWT:a}-\ref{PWT:b},~\ref{PWT:c}-\ref{PWT:d}, \mbox{and~\ref{PWT:g}-\ref{PWT:h})} in
	Fig.~\ref{algo:par-weight} takes 
	$\mathcal{O} \left ( \frac{n}{\mathcal{P}} \right )$ time.  Although the two \texttt{for}
	loops in Lines~\ref{PWT:c} and~\ref{PWT:i} are nested, the total number of weights updated
	are~$\mathcal{O} \left(n\right)$.  Thus computing the weights 
	in Step~$4$ takes~$\mathcal{O} \left ( \frac{n}{\mathcal{P}} \right )$ time.
	The \texttt{while} loop (Lines~\ref{PINEQ:a}-\ref{PINEQ:b}) in 
	Fig.~\ref{algo:par-eg-inequality} takes~$\mathcal{O} \left 
	( \frac{\mathcal{C}}{\mathcal{P}} \right )$ time.	 Therefore, the Erd\H{o}s-Gallai 
	inequalities in Step~$5$ are tested in $\mathcal{O} \left 
	( \frac{\mathcal{C}}{\mathcal{P}} \right )$ time.  Thus, the time complexity of the 
	algorithm is $\mathcal{O} \left ( \frac{n}{\mathcal{P}} + \frac{\mathcal{C}}{\mathcal{P}} + 
	\log \mathcal{P} \right ) = \mathcal{O} \left ( \frac{n}{\mathcal{P}} + \log \mathcal{P} \right )$.
\end{proof}

\vspace{0.01in}
\begin{theorem}
	The space complexity of the parallel algorithm for checking the
	Erd\H{o}s-Gallai characterization is $\mathcal{O} \left ( n \right )$.
\end{theorem}
\begin{center}
	\begin{minipage}{0.48\textwidth}
		\small{
			\noindent\fbox{%
				\begin{minipage}{\dimexpr0.98\linewidth-2\fboxsep-2\fboxrule\relax}
					\begin{algorithmic}[1]
						\State $d_0 \leftarrow n-1$
						\vspace{0.1in}
						\Statex \texttt{$\triangleright$ Initialize the weight array}
						\State \textbf{for} $i=1$ \textit{to} $n$ \textbf{\textit{in parallel} do} \label{PWT:a}
						\State \hspace{0.1in} $w_i \leftarrow 0$ \label{PWT:b}
						\vspace{0.1in}
						\Statex \texttt{$\triangleright$ Compute the weight values}
						\State \textbf{for} $i=1$ \textit{to} $n$ \textbf{\textit{in parallel} do} \label{PWT:c}
						\State \hspace{0.1in} \textbf{if} $d_i < d_{i-1}$ \textbf{then}
						\State \hspace{0.25in} \textbf{for} $j=d_{i-1}$ \textit{downto} $d_i+1$ \textbf{\textit{in parallel} do} \label{PWT:i}
						\State \hspace{0.4in} \textbf{if} $i-1 > w_j$ \textbf{then} \label{PWT:e}
						\State \hspace{0.55in} $w_j \leftarrow i-1$						
						\State \hspace{0.25in} \textbf{if} $i > w_{d_i}$ \textbf{then} \label{PWT:f}
						\State \hspace{0.4in} $w_{d_i} \leftarrow i$ \label{PWT:d}
						\vspace{0.1in}
						\Statex \texttt{$\triangleright$ Compute the larger weight values}
						\State \textbf{for} $j=d_n$ \textit{downto} $1$ \textbf{\textit{in parallel} do} \label{PWT:g}
						\State \hspace{0.1in} $w_j \leftarrow n$ \label{PWT:h}
					\end{algorithmic}
				\end{minipage}
			}
		}
		\captionof{figure}{Compute the weights in parallel.}
		\label{algo:par-weight}
	\end{minipage}
\end{center}
\begin{center}
	\begin{minipage}{0.48\textwidth}
		\small{
			\noindent\fbox{%
				\begin{minipage}{\dimexpr0.98\linewidth-2\fboxsep-2\fboxrule\relax}
					\begin{algorithmic}[1]						
						\State $k \leftarrow$ core id
						\State $flag \leftarrow TRUE$ \hspace{0.6in} \texttt{// shared variable}
						\vspace{0.1in}
						\Statex \texttt{$\triangleright$ Each core $\mathbb{P}_k$ executes the following in parallel:}
						\State $i \leftarrow k+1$
						\State \textbf{while} $i <= \mathcal{C}$ \textit{and} $flag = TRUE$ \textbf{do} \label{PINEQ:a}
						\State \hspace{0.1in} \textbf{if} $i \le w_{i}$ \textbf{then}
						\State \hspace{0.25in} \textbf{if} $H_i > i(i-1) + i(w_i-i) + H_n - H_{w_i}$ \textbf{then}
						\State \hspace{0.4in} $flag \leftarrow FALSE$
						\State \hspace{0.1in} \textbf{else if} $H_i > i(i-1) + H_n - H_i$ \textbf{then}
						\State \hspace{0.25in} $flag \leftarrow FALSE$
						\State \hspace{0.1in} $i \leftarrow i + \mathcal{P}$ \label{PINEQ:b}
						\vspace{0.1in}
						\State \textsc{\textbf{Omp-Barrier}}
						\State \textbf{return} $flag$
					\end{algorithmic}
				\end{minipage}
			}
		}
		\captionof{figure}{Check the Erd\H{o}s-Gallai inequalities in parallel.}
		\label{algo:par-eg-inequality}
	\end{minipage}
\end{center}	
\begin{proof}
	Storing the degree sequence	and the prefix sum of the degrees take $\mathcal{O}
	\left ( n \right )$ space.
\end{proof}

\begin{table}[b!]
	\caption{Data sets used in the experiments, where $n$, $m$, $\frac{2m}{n}$, and
		$\mathcal{C}$ denote the number of vertices, number of edges, average degree,
		and the corrected Durfee number of the networks, respectively.  
		\texttt{M} and \texttt{B} denote millions and billions, respectively.}
	\label{table:data-erdos-gallai}
	\ra{1.3}
	\centering
	\begin{tabular}{@{}llcccc@{}}
		\toprule[1.3pt]
		\multicolumn{1}{@{}l@{}}{\textbf{Network}} &
		\multicolumn{1}{@{}l@{}}{\textbf{~~Type}} &
		\multicolumn{1}{@{}c@{}}{$n$} &
		\multicolumn{1}{@{}c@{}}{$m$} &
		\multicolumn{1}{@{}c@{}}{$\frac{2m}{n}$} &
		\multicolumn{1}{@{}c@{}}{$\mathcal{C}$}
		\\
		\midrule[0.8pt]
		Friendster~\cite{SNAP} & Social & \texttt{65.6M} & \texttt{1.8B} & \texttt{55.06} & \texttt{2959} \\
		Twitter~\cite{SNAP} & Social & \texttt{40.56M} & \texttt{667.7M} & \texttt{32.93} & \texttt{6842} \\
		Los Angeles (LA)~\cite{socialContact} & Contact & \texttt{16.23M} & \texttt{459.3M} & \texttt{56.59} & \texttt{380} \\
		New York (NY)~\cite{socialContact} & Contact & \texttt{17.88M} & \texttt{480.1M} & \texttt{53.70} & \texttt{387} \\
		LiveJournal (LJ)~\cite{SNAP} & Social & \texttt{4.80M} & \texttt{42.85M} & \texttt{17.68} & \texttt{990} \\
		SmallWorld (SW)~\cite{smallWorld} & Random & \texttt{4.80M} & \texttt{48.00M} & \texttt{20.00} & \texttt{31} \\
		\bottomrule[1.3pt]
	\end{tabular}
\end{table}

\subsection{Performance Evaluation}
In this section, we present the data sets used in the experiments and
the strong scaling and runtime of our parallel algorithm for checking
the Erd\H{o}s-Gallai characterization.  We use the same experimental
setup as described before in Section~\ref{sec:experiment-graph-generate}.

\textbf{Data Sets.}  We use degree sequences of both artificial and real-world
networks for the experiments.  A summary of the networks is given in 
Table~\ref{table:data-erdos-gallai}.  
Friendster, Twitter, and LiveJournal~(LJ) are real-world online social
networks~\cite{SNAP}.  New York~(NY) and Los Angeles~(LA) are synthetic, yet
realistic social contact networks~\cite{socialContact}.  The SmallWorld
random network follows the Watts-Strogatz small world network model~\cite{smallWorld}.
Table~\ref{table:data-erdos-gallai} also shows that the corrected Durfee number 
$\mathcal{C}$ is significantly smaller than the number of vertices $n$ for all
six networks.

\textbf{Strong Scaling.}
The strong scaling and runtime of the parallel algorithm are illustrated
in Figs.~\ref{fig:strong_scaling_erdos_gallai}~and~\ref{fig:runtime_erdos_gallai},
respectively.  The speedup increases almost linearly with the increase in the
number of cores.  We observe better speedups for the degree sequences
of larger graphs and achieve a maximum speedup of 23 with 32
cores on the Friendster graph.

\section{Conclusion}
\label{sec:conclusion}
We presented an efficient 
parallel algorithm for generating random graphs
with prescribed degree sequences.  It can be used in studying various 
structural properties of and dynamics over a network, sampling graphs 
uniformly at random from the graphical realizations of a given degree
sequence and estimating the number of possible graphs with a given degree
sequence.  The algorithm never gets stuck, can generate every possible graph
with a positive probability, and exhibits good speedup.  We also compared
several important structural properties of the random graphs generated by
our parallel algorithm with that of the real-world graphs and random graphs
generated by the edge swapping method.  In addition, we developed an efficient
parallel algorithm for checking the  Erd\H{o}s-Gallai characterization
of a given degree sequence.  This algorithm can be of independent interest and
prove useful in parallelizing many other graph theory problems.  We believe the
parallel algorithms will contribute significantly in analyzing and mining
emerging complex systems and understanding interesting characteristics of
such networks.

\bibliographystyle{IEEEtran}
\bibliography{IEEEabrv,references}

\end{document}